\documentclass[10pt]{article}
\usepackage{graphicx,amsmath,amsfonts,bm,cite,epsfig,epsf,url,color}
\usepackage{fullpage}
\usepackage[small,bf]{caption}
\usepackage{times}
\newtheorem{theorem}{Theorem}[section]
\newtheorem{lemma}[theorem]{Lemma}

\newtheorem{proposition}[theorem]{Proposition}
\newtheorem{definition}{Definition}[section]

\newcommand{\R}{\mathbb{R}}

\newcommand{\<}{\langle}
\renewcommand{\>}{\rangle}

\renewcommand{\P}{\operatorname{\mathbb{P}}}
\newcommand{\E}{\operatorname{\mathbb{E}}}

\newcommand{\RO}{{\cal R}_{\Omega}}
\newcommand{\ROs}[1]{{\cal R}_{\Omega_{#1}}}

\newcommand{\PT}{{\cal P}_T}
\newcommand{\PTc}{{\cal P}_{T^\perp}}

\newcommand{\Id}{\bm I}
\newcommand{\OpId}{\mathcal{I}}

\newcommand{\eab}{\vct{e}_{a}\vct{e}_{b}^*}
\newcommand{\ecd}{\vct{e}_{c}\vct{e}_{d}^*}
\newcommand{\eabk}{\vct{e}_{a_k}\vct{e}_{b_k}^*}

\newcommand{\vct}[1]{\bm{#1}}
\newcommand{\mtx}[1]{\bm{#1}}

\newcommand{\trace}{\operatorname{Tr}}
\newcommand{\lspan}[1]{\operatorname{span}{#1}}

\numberwithin{equation}{section}

\def \endprf{\hfill {\vrule height6pt width6pt depth0pt}\medskip}
\newenvironment{proof}{\noindent {\bf Proof} }{\endprf\par}

\newcommand{\eq}[1]{(\ref{eq:#1})}

\title{A Simpler Approach to Matrix Completion}

\author{Benjamin Recht\\
  \vspace{-.1cm}\\
{\normalsize Department of Computer Sciences, University of Wisconsin-Madison}\\
 {\normalsize 1210 W Dayton St, Madison, WI 53706}\\
 {\normalsize email: {\tt brecht@cs.wisc.edu}}}

\date{October 2009}

\begin{document}

\maketitle

\vspace{-0.3in}

\begin{abstract}
This paper provides the best bounds to date on the number of randomly sampled entries required to reconstruct an unknown low rank matrix.  These results improve on prior work by Cand\`es and Recht~\cite{CandesRecht08}, Cand\`es and Tao~\cite{CandesTao09}, and Keshavan, Montanari, and Oh~\cite{Keshavan09}.  The reconstruction is accomplished by minimizing the nuclear norm, or sum of the singular values, of the hidden matrix subject to agreement with the provided entries. If the underlying matrix satisfies a certain incoherence condition, then the number of entries required is equal to a quadratic logarithmic factor times the number of parameters in the singular value decomposition.  The proof of this assertion is short, self contained, and uses very elementary analysis.  The novel techniques herein are based on recent work in quantum information theory.
\end{abstract}


\section{Introduction}
\label{sec:intro}

Recovering a low rank matrix from a given subset of its entries is a recurring problem in collaborative filtering~\cite{Rennie05}, dimensionality reduction~\cite{Linial95, Weinberger06}, and multi-class learning~\cite{Argyriou08,Obozinski09}.  While a variety of heuristics have been developed across many disciplines, the general problem of finding the lowest rank matrix satisfying equality constraints is NP-hard. All known algorithms which can compute the lowest rank solution for all instances require time at least exponential in the dimensions of the matrix in both theory and practice \cite{Grigoriev84}.
 
In sharp contrast to such worst case pessimism, Cand\`es and Recht showed that most low rank matrices could be recovered from  most sufficiently large sets of entries by computing the matrix of minimum \emph{nuclear norm} that agreed with the provided entries~\cite{CandesRecht08}, and furthermore the revealed set of entries could comprise a vanishing fraction of the entire matrix.  The nuclear norm is equal to the sum of the singular values of a matrix and is the best convex lower bound of the rank function on the set of matrices whose singular values are all bounded by $1$. The intuition behind this heuristic is that whereas the rank function counts the number of nonvanishing singular values, the nuclear norm sums their amplitude, much like how the $\ell_1$ norm is a useful surrogate for counting the number of nonzeros in a vector.   Moreover, the nuclear norm can be minimized subject to equality constraints via semidefinite programming.

Nuclear norm minimization had long been observed to produce very low-rank solutions in practice (see, for example~\cite{Beck98,FazelThesis,Fazel01,SrebroThesis,Mesbahi97}), but only very recently was there any theoretical basis for when it produced the minimum rank solution.  The first paper to provide such foundations was~\cite{Recht07}, where Recht, Fazel, and Parrilo developed probabilistic techniques to study average case behavior and showed that the nuclear norm heuristic could solve most instances of the rank minimization problem assuming the number of linear constraints was sufficiently large. The results in \cite{Recht07}~inspired a groundswell of interest in theoretical guarantees for rank minimization, and these results lay the foundation for~\cite{CandesRecht08}.  Cand\`es and Recht's bounds were subsequently improved by Cand\`es and Tao~\cite{CandesTao09} and Keshavan, Montanari, and Oh~\cite{Keshavan09} to show that one could, in special cases, reconstruct a low-rank matrix by observing a set of entries of size at most a polylogarithmic factor larger than the intrinsic dimension of the variety of rank $r$ matrices.   

This paper sharpens the results in~\cite{CandesTao09,Keshavan09} to provide a bound on the number of entries required to reconstruct a low rank matrix which is optimal up to a small numerical constant and one logarithmic factor.  The main theorem makes minimal assumptions about the low rank matrix of interest.  Moreover, the proof is very short and relies on mostly elementary analysis.  

In order to precisely state the main result, we need one definition. Cand\`es and Recht observed that it is impossible to recover a matrix which is equal to zero in nearly all of its entries unless all of the entries of the matrix are observed (consider, for example, the rank one matrix which is equal to $1$ in one entry and zeros everywhere else). In other words, the matrix cannot be mostly equal to zero on the observed entries.  This motivated the following definition
\begin{definition}
\label{def:coherence} Let $U$ be a subspace of $\mathbb{R}^n$ of
dimension $r$ and $\mtx{P}_U$ be the orthogonal projection onto $U$.
Then the \emph{coherence} of $U$ (vis-\`a-vis the standard basis
$(\vct{e}_i)$) is defined to be
\begin{equation}
\label{eq:coherence} \mu(U) \equiv \frac{n}{r} \max_{1 \le i \le n}
\|\mtx{P}_U
  \vct{e}_i\|^2.
\end{equation}
\end{definition}
 Note that for any subspace, the smallest $\mu(U)$ can be is $1$, achieved, for example, if $U$ is spanned by vectors whose entries all have magnitude $1/\sqrt{n}$. The largest possible value for $\mu(U)$ is $n/r$ which would correspond to any subspace that contains a standard basis element. If a matrix has row and column spaces with low coherence, then each entry can be expected to provide about the same amount of information.

Recall that the \emph{nuclear norm} of an $n_1\times n_2$ matrix $\mtx{X}$ is the sum of the singular values of $\mtx{X}$,  $\|\mtx{X}\|_* = \sum_{k = 1}^{\min\{n_1,n_2\}} \sigma_k(\mtx{X})$, where, here and below, $\sigma_k(\mtx{X})$ denotes the $k$th largest singular value of $\mtx{X}$.  The main result of this paper is the following
\begin{theorem} \label{thm:main} 
Let $\mtx{M}$ be an $n_1 \times n_2$ matrix of  rank $r$ with singular value decomposition $\mtx{U}\mtx{\Sigma}\mtx{V}^*$.   Without loss of generality, impose the conventions $n_1\leq n_2$, $\mtx{\Sigma}$ is $r\times r$, $\mtx{U}$ is $n_1 \times r$ and $\mtx{V}$ is $n_2 \times r$. Assume that
\begin{description}
\item[{A0}] The row and column spaces have coherences bounded above by some positive $\mu_0$.
\item[{A1}] The matrix $\mtx{U}\mtx{V}^*$ has a maximum entry bounded by $\mu_1
  \sqrt{r/(n_1 n_2)}$ in absolute value for some positive $\mu_1$.
\end{description}
Suppose $m$ entries of $\mtx{M}$ are observed with locations sampled uniformly at random. Then  if
\begin{equation}\label{eq:main3}
  m \geq \, 32 \max\{\mu_1^2, \mu_0\} \, r(n_1+n_2) \,\,\beta \log^2 (2n_2)
\end{equation}
for some $\beta>1$, the minimizer to the problem 
\begin{equation}
 \label{eq:sdp}
  \begin{array}{ll}
\textrm{minimize}   & \quad \|\mtx{X}\|_*\\
\textrm{subject to} & \quad   X_{ij} = M_{ij} \quad (i,j) \in \Omega.
 \end{array}
\end{equation}
is unique and equal to $\mtx{M}$ with probability at least $1-6\log(n_2) (n_1+n_2)^{2-2\beta}-n_2^{2-2\beta^{1/2}}$.  
\end{theorem}

The assumptions $\bf{A0}$ and $\bf{A1}$ were introduced in~\cite{CandesRecht08}.  Both  $\mu_0$ and $\mu_1$ may depend on $r$, $n_1$, or $n_2$. Moreover, note that $\mu_1 \leq \mu_0 \, \sqrt{r}$ by the Cauchy-Schwarz inequality.  As  shown in~\cite{CandesRecht08}, both subspaces selected from the uniform distribution and spaces constructed as the span of singular vectors with bounded entries are not only incoherent with the standard basis, but also obey {\bf A1} with high probability for values of $\mu_1$ at most logarithmic in $n_1$ and/or $n_2$.  Applying this theorem to the models studied in Section 2 of~\cite{CandesRecht08}, we find that there is a numerical constant $c_u$ such that $c_u r (n_1+n_2)\log^5(n_2)$ entries are sufficient to reconstruct a rank $r$ matrix whose row and column spaces are sampled from the Haar measure on the Grassmann manifold.  If $r>\log(n_2)$, the number of entries can be reduced to $c_u r (n_1+n_2)\log^4(n_2)$.  Similarly, there is a numerical constant $c_i$ such that $c_i \mu_0^2 r(n_1+n_2) \log^3(n_2)$ entries are sufficient to recover a matrix of arbitrary rank $r$ whose singular vectors have entries with magnitudes bounded by $\sqrt{\mu_0/n_1}$.  

Theorem~\ref{thm:main} greatly improves upon prior results.  First of all, it has the weakest  assumptions on the matrix to be recovered. In addition to assumption {\bf A1}, Cand\`es and Tao require a ``strong incoherence condition'' (see~\cite{CandesTao09}) which is considerably more restrictive than the assumption {\bf A0} in Theorem~\ref{thm:main}.  Many of their results also require restrictions on the rank of $\mtx{M}$, and their bounds depend superlinearly on $\mu_0$. Keshavan \emph{et al} require the matrix rank to be no more than $\log(n_2)$, and require bounds on the maximum magnitude of the  entries in $\mtx{M}$ and the ratios $\sigma_1(\mtx{M})/\sigma_r(\mtx{M})$ and $n_2/n_1$.   Theorem~\ref{thm:main} makes no such assumptions about the rank, aspect ratio, nor condition number of $\mtx{M}$.  Moreover, \eq{main3} has a smaller log factor than~\cite{CandesTao09}, and features numerical constants that are both explicit and small.  

Also note that there is not much room for improvement in the bound for $m$. It is a consequence of the coupon collector's problem that at least  $n_2 \log n_2$ uniformly sampled entries are necessary just to guarantee that at least one entry in every row and column is observed with high probability. In addition, rank $r$ matrices have $r(n_1+n_2-r)$ parameters, a fact that can be verified by counting the number of degrees of freedom in the singular value decomposition.  Interestingly, Cand\`es and Tao showed that $C\mu_0 n_2 r \log(n_2)$ entries were \emph{necessary} for completion when the entries are sampled uniformly at random~\cite{CandesTao09}.   Hence,~\eq{main3} is optimal up to a small numerical constant times $\log(n_2)$.

Most importantly, the proof of Theorem~\ref{thm:main} is short and straightforward. Cand\`es and Recht employed sophisticated tools from the study of random variables on Banach spaces including decoupling tools and powerful moment inequalities for the norms of random matrices. Cand\`es and Tao rely on intricate moment calculations spanning over $30$ pages.   The present work only uses basic matrix analysis, elementary large deviation bounds, and a noncommutative version of Bernstein's Inequality proven here in the Appendix.

The proof of Theorem~\ref{thm:main} is inspired by a recent paper in quanutm information which considered the problem of reconstructing the density matrix of a quantum ensemble using as few measurements as possible~\cite{Gross09}.  Their work adapted results from~\cite{CandesRecht08} and~\cite{CandesPlan09} to the quantum regime by using special algebraic properties of quantum measurements.  Their proof followed a methodology analogous to the approach of Cand\`es and Recht but had two main differences: they used a sampling with replacement model as a proxy for uniform sampling, and they deployed a powerful noncommutative Chernoff bound developed by Ahlswede and Winter for use in quantum information theory~\cite{AhlswedeWinter02}.  In this paper, I adapt these two strategies from~\cite{Gross09} to the matrix completion problem. In section~\ref{sec:swr} I show how the sampling with replacement model bounds probabilities in the uniform sampling model, and present very short proofs of some of the main results in~\cite{CandesRecht08}.   Surprisingly, this yields a simple proof of Theorem~\ref{thm:main}, provided in Section~\ref{sec:main-result}, which has the least restrictive assumptions of any assertion proven thus far.

\section{Preliminaries and notation}\label{sec:notation}

Before continuing, let us survey the notations used throughout the paper. I closely follow the conventions established in~\cite{CandesRecht08}, and invite the reader to consult this reference for a more thorough discussion of the matrix completion problem and the associated convex geometry.  A thorough introduction to the necessary matrix analysis used in this paper can be found in~\cite{Recht07}.

Matrices are bold capital, vectors are bold lowercase and scalars or entries are not bold. For example, $\mtx{X}$ is a matrix, and $X_{ij}$ its $(i,j)$th entry. Likewise $\vct{x}$ is a vector, and $x_i$ its $i$th component. If $\vct{u}_k\in \R^n$ for $1\leq k \leq d$ is a collection of vectors,  $[\vct{u}_1,\ldots,\vct{u}_d]$ will denote the $n\times d$ matrix whose $k$th column is $\vct{u}_k$.  $\vct{e}_k$ will denote the $k$th standard basis vector in $\R^d$, equal to $1$ in component $k$ and $0$ everywhere else.  The dimension of $\vct{e}_k$ will always be clear from context.  $\mtx{X}^*$ and $\vct{x}^*$ denote the transpose of matrices $\mtx{X}$ and vectors $\vct{x}$ respectively.

A variety of norms on matrices will be discussed. The spectral norm of a matrix is denoted by $\|\mtx{X}\|$.  The Euclidean inner product between two matrices is $\<\mtx{X}, \mtx{Y}\> = \trace(\mtx{X}^* \mtx{Y})$, and the corresponding Euclidean norm, called the Frobenius or Hilbert-Schmidt norm, is denoted $\|\mtx{X}\|_F$.  That is, $\|\mtx{X}\|_F=\<\mtx{X},\mtx{X}\>^{1/2}$.  The nuclear norm of a
matrix $\mtx{X}$ is $\|\mtx{X}\|_*$.  The maximum entry of $\mtx{X}$ (in absolute value) is denoted by $\|\mtx{X}\|_\infty \equiv \max_{ij} |X_{ij}|$. For vectors, the only norm applied is the usual Euclidean $\ell_2$ norm, simply denoted as $\|\vct{x}\|$.

Linear transformations that act on matrices will be denoted by calligraphic letters.  In particular, the identity operator will be denoted by $\OpId$. The spectral norm (the top singular value) of such an operator will be denoted by $\|{\cal A}\| = \sup_{\mtx{X} :\|\mtx{X}\|_F \le 1} \, \|{\cal A}(\mtx{X})\|_F$.

Fix once and for all a matrix $\mtx{M}$ obeying the assumptions of Theorem~\ref{thm:main}. Let $\mtx{u}_k$ (respectively $\mtx{v}_k$) denote the $k$th column of $\mtx{U}$ (respectively $\mtx{V}$). Set $U \equiv \lspan{(\vct{u}_1, \ldots, \vct{u}_r)}$, and $V \equiv \lspan{(\vct{v}_1, \ldots, \vct{v}_r)}$.  Also assume, without loss of generality, that $n_1 \leq n_2$. It is convenient to introduce the orthogonal decomposition $\R^{n_1 \times n_2} = T \oplus T^\perp$ where $T$ is the linear space spanned by elements of the form $\vct{u}_k \vct{y}^*$ and $\vct{x} \vct{v}_k^*$, $1 \le k \le r$, where $\vct{x}$ and $\vct{y}$ are arbitrary, and $T^\perp$ is its orthogonal complement.  $T^\perp$ is the subspace of matrices spanned by the family $(\vct{x}\vct{y}^*)$, where $\vct{x}$ (respectively $\vct{y}$) is any vector orthogonal to $U$ (respectively $V$).

The orthogonal projection $\PT$ onto $T$ is given by
\begin{equation}\label{eq:PT}
  \PT(\mtx{Z}) = \mtx{P}_U \mtx{Z} + \mtx{Z} \mtx{P}_V  - \mtx{P}_U \mtx{Z} \mtx{P}_V,
\end{equation}
where $\mtx{P}_{U}$ and $\mtx{P}_V$ are the orthogonal projections onto $U$ and $V$ respectively.  Note here that while $\mtx{P}_{U}$ and $\mtx{P}_V$ are matrices, $\PT$ is a linear operator mapping
matrices to matrices.  The orthogonal projection onto $T^\perp$ is given by
\[
\PTc(\mtx{Z}) = (\OpId - \PT)(\mtx{Z}) = (\Id_{n_1} -
\mtx{P}_{U}) \mtx{Z} (\Id_{n_2} - \mtx{P}_{V})
\]
where $\Id_d$ denotes the $d\times d$ identity matrix.  It follows from the definition \eq{PT} of $\PT$ that
\[
  \PT(\eab) = (\mtx{P}_U \vct{e}_a) \vct{e}_b^* + \vct{e}_a (\mtx{P}_V \vct{e}_b)^* - (\mtx{P}_U \vct{e}_a)(\mtx{P}_V \vct{e}_b)^*.
\]
This gives
\[
\|\PT(\eab)\|_F^2 = \<\PT(\eab), \eab\>   = \|\mtx{P}_U
  \vct{e}_a\|^2 + \|\mtx{P}_V
  \vct{e}_b\|^2 - \|\mtx{P}_U \vct{e}_a\|^2 \, \|\mtx{P}_V \vct{e}_b\|^2\,.
\]
Since $\|\mtx{P}_U \vct{e}_a\|^2 \le \mu(U)r/n_1$ and
$\|\mtx{P}_V \vct{e}_b\|^2 \le \mu(V)r/n_2$,
\begin{equation}
  \|\PT(\eab)\|_F^2 \le  \max\{\mu(U),\mu(V)\} r\frac{n_1+n_2}{n_1n_2} \leq  \mu_0 r\frac{n_1+n_2}{n_1n_2}
\label{eq:PTea-bound}
\end{equation}
I will make frequent use of this calculation throughout the sequel.

\section{Sampling with Replacement}\label{sec:swr}

As discussed above, the main contribution of this work is an analysis of uniformly sampled sets of entries via the study of a sampling with replacement model.  All of the previous work~\cite{CandesRecht08,CandesTao09,Keshavan09} studied a Bernoulli sampling model as a proxy for uniform sampling.  There, each entry was revealed independently with probability equal to $p$.  In all of these results, the theorem statements concerned sampling sets of $m$ entries uniformly, but it was shown that probability of failure under Bernoulli sampling with $p=\tfrac{m}{n_1n_2}$ closely approximated the probability of failure under uniform sampling.  The present work will analyze the situation where each entry index is sampled independently from the uniform distribution on $\{1,\ldots, n_1\} \times\{1,\ldots,n_2\}$.  This modification of the sampling model gives rise to all of the simplifications below.  

It would appear that sampling with replacement is not suitable for analyzing matrix completion as one might encounter duplicate entries.  However, just as is the case with Bernoulli sampling, bounding the likelihood of error when sampling with replacement allows us to bound the probability of the nuclear norm heuristic failing under uniform sampling.
\begin{proposition}\label{prop:swr}
	The probability that the nuclear norm heuristic fails when the set of observed entries is sampled uniformly from the collection of sets of size $m$ is less than or equal to the probability that the heuristic fails when $m$ entries are sampled independently  with replacement.
\end{proposition}
\begin{proof} The proof follows the argument in Section II.C of~\cite{crt06}. Let $\Omega'$ be a collection of $m$ entries, each sampled independently from the uniform distribution on $\{1,\ldots, n_1\} \times\{1,\ldots,n_2\}$.  Let $\Omega_k$ denote a set of entries of size $k$ sampled uniformly from all collections of entries of size $k$.  It follows that
\begin{align*}
	\P(\mbox{Failure}(\Omega')) &= \sum_{k=0}^m P(\mbox{Failure}(\Omega') ~|~|\Omega'|=k)
		P(|\Omega'|=k)\\
	&= \sum_{k=0}^m P(\mbox{Failure}(\Omega_k))P(|\Omega'|=k)\\
	&\geq P(\mbox{Failure}(\Omega_m))\sum_{k=0}^m P(|\Omega'|=k)= P(\mbox{Failure}(\Omega_m))\,.
\end{align*}
Where the inequality follows because $P(\mbox{Failure}(\Omega_m)) \geq P(\mbox{Failure}(\Omega_{m'}))$ if $m\leq m'$.  That is, the probability decreases as the number of entries revealed is increased.
\end{proof}
 
Surprisingly, changing the sampling model makes most of the theorems from~\cite{CandesRecht08} simple consequences of a noncommutative variant of Bernstein's Inequality.  

\begin{theorem}[Noncommutative Bernstein Inequality]\label{thm:bernstein}
Let $\mtx{X}_1,\ldots,\mtx{X}_L$ be independent zero-mean random matrices of dimension $d_1\times d_2$.  Suppose $\rho_k^2=\max\{\|\E[\mtx{X}_k\mtx{X}_k^*]\|,\|\E[\mtx{X}_k^*\mtx{X}_k]\| \}$ and $\|\mtx{X}_k\|\leq M$ almost surely for all $k$.  Then for any $\tau>0$,
\[
	\P\left[ \left\| \sum_{k=1}^L \mtx{X}_k\right\| > \tau \right]  \leq (d_1+d_2)\exp\left( \frac{-\tau^2/2}
	{\sum_{k=1}^L \rho_k^2 + M\tau/3} \right)\,.
\]
\end{theorem}
Note that in the case that $d_1=d_2=1$, this is precisely the two sided version of the standard Bernstein Inequality.  When the $\mtx{X}_k$ are diagonal, this bound is the same as applying the standard Bernstein Inequality and a union bound to the diagonal of the matrix summation.  Furthermore, observe that the right hand side is less than $(d_1+d_2) \exp(-\tfrac{3}{8}\tau^2/( \sum_{k=1}^L\rho_k^2))$ as long as $\tau\leq \tfrac{1}{M} \sum_{k=1}^L\rho_k^2$.   This condensed form of the inequality will be used exclusively throughout.  Theorem~\ref{thm:bernstein} is a corollary of an Chernoff bound for finite dimensional operators developed by Ahlswede and Winter~\cite{AhlswedeWinter02}. A similar inequality for symmetric i.i.d. matrices is proposed in~\cite{Gross09}.  The proof is provided in the Appendix.

Let us now record two theorems, proven for the Bernoulli model in~\cite{CandesRecht08}, that admit very simple proofs in the sampling with replacement model.   The theorem statements requires some additional notation.  Let $\Omega = \{(a_k,b_k)\}_{k=1}^l$ be a collection of indices sampled uniformly with replacement.  Set $\RO$ to be the operator
\[
	\RO(\mtx{Z}) = \sum_{k=1}^{|\Omega|} \langle \eabk, \mtx{Z} \rangle \eabk\,.
\]
Note that the $(i,j)$th component of $\RO(\mtx{X})$ is zero unless $(i,j)\in\Omega$.  For $(i,j)\in \Omega$, $\RO(\mtx{X})$ is equal to $X_{ij}$ times the multiplicity of $(i,j) \in \Omega$. Unlike in previous work on matrix completion, $\RO$ is not a projection operator if there are duplicates in $\Omega$.  Nonetheless, this does not adversely affect the argument, and $\RO(\mtx{X})=0$ if and only if $X_{ab}=0$ for all $(a,b)\in\Omega$.  Moreover, we can show that the maximum duplication of any entry is always less than $\tfrac{8}{3}\log(n_2)$ with very high probability.

\begin{proposition}\label{prop:duplicate-count}
With probability at least $1-n_2^{2-2\beta}$, the maximum number of repetitions of any entry in $\Omega$ is less than $\tfrac{8}{3} \beta\log(n_2)$ for $n_2\geq 9$ and $\beta>1$.
\end{proposition}

\begin{proof}
This assertion can be proven by applying a standard Chernoff bound for the Bernoulli distribution.  Note that for a fixed entry, the probability it is sampled more than $t$ times is equal to the probability of more than $t$ heads occurring in a sequence of $m$ tosses where the probability of a head is $\tfrac{1}{n_1n_2}$.  This probability can be upper bounded by
\[
	\P[\mbox{more than}~t~\mbox{heads in}~m~\mbox{trials}] \leq \left(\frac{m}{n_1n_2 t}\right)^t \exp\left(t-\frac{m}{n_1n_2}\right)
\]
(see~\cite{Hagerup90}, for example).  Applying the union bound over all of the $n_1n_2$ entries and the fact that $\tfrac{m}{n_1n_2}<1$, we have
\begin{align*}
\P[\mbox{any entry is selected more than}~\tfrac{8}{3} \beta\log(n_2)~\mbox{times}]
&\leq n_1n_2 \left(\tfrac{8}{3}\beta \log(n_2) \right)^{-\tfrac{8}{3} \beta\log(n_2)} \exp\left(\tfrac{8}{3} \beta\log(n_2)\right) \leq
n_2^{2-2\beta}
\end{align*}
when $n_2\geq 9$.
\end{proof}

This application of the Chernoff bound is very crude, and much tighter bounds can be derived using more careful analysis.  For example in~\cite{Gonnet81}, the maximum oversampling is shown to be bounded by $O(\tfrac{\log(n_2)}{\log\log(n_2)})$.  For our purposes here, the loose upper bound provided by Proposition~\ref{prop:duplicate-count} will be more than sufficient.

In addition to this bound on the norm of $\RO$, the following theorem asserts that the operator $\PT\RO\PT$ is also very close to an isometry on $T$ if the number of sampled entries is sufficiently large.  This result is analgous to the Theorem 4.1 in~\cite{CandesRecht08} for the Bernoulli model, whose proof uses several powerful theorems from the study of probability in Banach spaces.  Here, one only needs to compute a few low order moments and then apply Theorem~\ref{thm:bernstein}.

\begin{theorem}
 \label{thm:near-isometry}
Suppose $\Omega$ is a set of entries of size $m$ sampled independently and uniformly with replacement.  Then for all $\beta>1$,
\[
	\frac{n_1n_2}{m} \left\|\PT \RO \PT - \frac{m}{n_1n_2}\PT\right\| \leq \sqrt{ \frac{16 \mu_0 r(n_1+n_2)\,\beta \log(n_2)}{3m}}
\]
with probability at least $1- 2n_2^{2-2\beta}$ provided that $m>\tfrac{16}{3} \mu_0 r(n_1+n_2)\,\beta \log(n_2)$.
\end{theorem}

\begin{proof}
Decompose any matrix $\mtx{Z}$ as $\mtx{Z} = \sum_{ab} \<\mtx{Z},
\eab\> \eab$ so that
\begin{equation}\label{eq:pt-def}
\PT(\mtx{Z}) = \sum_{ab} \<\PT(\mtx{Z}), \eab\> \eab = \sum_{ab}
\<\mtx{Z}, \PT(\eab)\> \eab.
\end{equation}
For $k=1,\ldots, m$ sample $(a_k,b_k)$ from $\{1,\ldots,n_1\}\times\{1,\ldots,n_2\}$ uniformly with replacement. Then $\RO \PT(\mtx{Z}) = \sum_{k=1}^m\, \<\mtx{Z},
\PT(\eabk)\> \, \eabk$ which gives
\[
(\PT \RO \PT)(\mtx{Z}) = \sum_{k=1}^m \, \<\mtx{Z},
\PT(\eabk)\> \,\PT(\eabk).
\]
Now the fact that the operator $\PT \RO \PT$ does not deviate from its expected value
\[
\E (\PT \RO \PT) = \PT  (\E \RO) \PT  = \PT (\frac{m}{n_1n_2} \OpId) \PT = \frac{m}{n_1n_2}\PT
\]
in the spectral norm can be proven using the Noncommutative Bernstein Inequality.  

To proceed, define the operator $\mathcal{T}_{ab}$ which maps $\mtx{Z}$ to $\langle \PT(\eab), \mtx{Z}\rangle \PT(\eab)$.  This operator is rank one, has operator norm $\|\mathcal{T}_{ab}\| = \|\PT(\eab)\|_F^2$, and we have $\PT = \sum_{a,b} \mathcal{T}_{ab}$ by~\eq{pt-def}. Hence, for $k=1,\ldots, m$,  $\E[ \mathcal{T}_{a_kb_k}]=\frac{1}{n_1n_2}\PT$.  

Observe that if $\mtx{A}$ and $\mtx{B}$ are positive semidefinite, we have $\|\mtx{A}-\mtx{B}\| \leq \max\{\|\mtx{A}\|,\|\mtx{B}\|\}$.  Using this fact, we can compute the bound
\[
\|\mathcal{T}_{a_kb_k}-\tfrac{1}{n_1n_2}\PT\| \leq \max\{\|\PT(\eabk)\|_F^2,\tfrac{1}{n_1n_2} \}\leq \mu_0 r\frac{n_1+n_2}{n_1n_2}\,,
\] 
where the final inequality follows from~\eq{PTea-bound}.  We also have
\begin{align*}
\|\E[ (\mathcal{T}_{a_kb_k}-\tfrac{1}{n_1n_2}\PT)^2 ]\|
&=\|\E[ \|\PT(\eabk)\|_F^2 \mathcal{T}_{a_kb_k}]-\frac{1}{n_1^2n_2^2}\PT ]\|\\
&\leq \max\{\|\E[ \|\PT(\eabk)\|_F^2 \mathcal{T}_{a_kb_k}]\|,\frac{1}{n_1^2n_2^2}\}\\
&\leq\max\{\|\E[ \mathcal{T}_{a_kb_k}]\| \, \mu_0 r\frac{n_1+n_2}{n_1n_2},\frac{1}{n_1^2n_2^2}\}  \leq \mu_0 r \frac{n_1 + n_2}{n_1^2n_2^2}
\end{align*}
The theorem now follows by applying the Noncommutative Bernstein Inequality.
\end{proof}

The next theorem is an analog of Theorem 6.3 in~\cite{CandesRecht08} or Lemma 3.2 in~\cite{Keshavan09}.  This theorem asserts that for a fixed matrix, if one sets all of the entries not in $\Omega$ to zero it remains close to a multiple of the original matrix in the operator norm.
\begin{theorem}
  \label{thm:inf-norm-upper-bound}
Suppose $\Omega$ is a set of entries of size $m$ sampled independently and uniformly with replacement and let $\mtx{Z}$ be a fixed $n_1 \times n_2$ matrix.  Assume without loss of generality that $n_1\leq n_2$,  Then for all $\beta>1$,
\[
	 \left\|\left(\frac{n_1n_2}{m}\RO-\OpId\right)(\mtx{Z})\right\|\leq \sqrt{\frac{8\beta n_1n_2^2\log(n_1+n_2)}{3m}} \|\mtx{Z}\|_\infty
\]
with probability at least $1- (n_1+n_2)^{1-\beta}$ provided that $m> 6\beta n_1 \log(n_1+n_2)$.
\end{theorem}

\begin{proof}
First observe that the operator norm can be upper bounded by a multiple of the matrix infinity norm 
\begin{align*}
	\|\mtx{Z}\| = \sup_{\stackrel{\|\vct{x}\|=1}{\|\vct{y}\|=1}}  \sum_{a,b} Z_{ab}y_ax_b \leq \left(\sum_{a,b} Z_{ab}^2y_a^2 \right)^{1/2} \left(\sum_{a,b} x_b^2\right)^{1/2}
	\leq \sqrt{n_2} \max_a \left(\sum_b Z_{ab}^2\right)^{1/2}\leq \sqrt{n_1n_2}\|\mtx{Z}\|_\infty
\end{align*}

Note that	 $\tfrac{n_1n_2}{m}\RO(\mtx{Z}) - \mtx{Z} = \frac{1}{m}\sum_{k=1}^m n_1n_2 Z_{a_kb_k}\eabk - \mtx{Z}$.  This is a sum of zero-mean random matrices, and $\|n_1n_2 Z_{a_kb_k}\eabk - \mtx{Z}\|\leq  \|n_1n_2 Z_{a_kb_k}\eabk\| + \|\mtx{Z}\|< \tfrac{3}{2} n_1n_2 \|\mtx{Z}\|_{\infty}$ for $n_1\geq 2$.  We also have
\begin{align*}
\left\| \E\left[(n_1n_2 Z_{a_kb_k}\eabk - \mtx{Z})^* (n_1n_2 Z_{a_kb_k}\eabk - \mtx{Z}) \right]\right\|
&=\left\| n_1n_2 \sum_{c,d}  Z_{cd}^2 \vct{e}_d\vct{e}_d^* - \mtx{Z}^*\mtx{Z}\right\|\\
&\leq \max\left\{ \left\| n_1n_2 \sum_{c,d}  Z_{cd}^2 \vct{e}_d\vct{e}_d^*\right\|, \left\|\mtx{Z}^*\mtx{Z}\right\|\right\}\\
&\leq n_1 n_2^2\|\mtx{Z}\|_\infty^2
\end{align*}
where we again use the fact that $\|\mtx{A}-\mtx{B}\|\leq \max\{\|\mtx{A}\|,\|\mtx{B}\|\}$ for positive semidefinite $\mtx{A}$ and $\mtx{B}$.  A similar calculation holds for $(n_1n_2 Z_{a_kb_k}\eabk - \mtx{Z}) (n_1n_2 Z_{a_kb_k}\eabk - \mtx{Z})^*$. The theorem now follows by the Noncommutative Bernstein Inequality.
\end{proof}

Finally, the following Lemma is required to prove Theorem~\ref{thm:main}.  Succinctly, it says that for a fixed matrix in $T$, the operator $\PT\RO$ does not increase the matrix infinity norm.
\begin{lemma}
  \label{lemma:norm-contraction}
Suppose $\Omega$ is a set of entries of size $m$ sampled independently and uniformly with replacement and let $\mtx{Z}\in T$ be a fixed $n_1 \times n_2$ matrix. Assume without loss of generality that $n_1\leq n_2$.  Then for all $\beta>2$,
    \[
	 \left\|\frac{n_1n_2}{m}\PT\RO(\mtx{Z}) - \mtx{Z}\right\|_\infty \leq \sqrt{\frac{8\beta \mu_0 r(n_1+n_2) \log n_2}{3m}} \|\mtx{Z}\|_\infty
\]
with probability at least $1- 2n_2^{2-\beta}$ provided that $m>\tfrac{8}{3}\beta \mu_0 r (n_1+n_2) \log n_2$.
\end{lemma}

\begin{proof}
	This lemma can be proven using the standard Bernstein Inequality.  For each matrix index $(c,d)$, sample $(a,b)$ uniformly at random to define the random variable $\xi_{cd} = \langle \ecd, n_1n_2 \langle \eab,\mtx{Z}\rangle\PT(\eab) - \mtx{Z} \rangle$.  We have $\E[\xi_{cd}]=0$, $|\xi_{cd}| \leq \mu_0 r (n_1+n_2) \|\mtx{Z}\|_\infty$, and
\begin{align*}
\E[\xi_{cd}^{2}] &=  \frac{1}{n_1n_2} \sum_{a,b} \langle \ecd, n_1n_2 \langle \eab,\mtx{Z}\rangle\PT(\eab) - \mtx{Z} \rangle^2\\
&=n_1n_2\sum_{a,b} \langle \PT(\ecd), \eab \rangle^2 \langle \eab,\mtx{Z}\rangle^2 - Z_{cd}^2\\
&\leq n_1n_2\|\PT(\ecd)\|^2_F \|\mtx{Z}\|_\infty^2\leq \mu_0 r (n_1+n_2) \|\mtx{Z}\|_\infty^2\,.
\end{align*}
Since the $(c,d)$ entry of $\frac{n_1n_2}{m}\PT\RO(\mtx{Z}) - \mtx{Z}$ is identically distributed to $\tfrac{1}{m}\sum_{k=1}^m \xi_{cd}^{(k)}$, where $\xi_{cd}^{(k)}$ are i.i.d. copies of $\xi_{cd}$, we have by Bernstein's Inequality and the union bound:
\[
	\Pr\left[  \left\|\frac{n_1n_2}{m}\PT\RO(\mtx{Z}) - \mtx{Z}\right\|_\infty > \sqrt{\frac{8\beta \mu_0 r (n_1+n_2) \log(n_2)}{3m}} \|\mtx{Z}\|_\infty \right] \leq 2n_1n_2 \exp(-\beta \log(n_2))\leq 2n_2^{2-\beta}\, .
\]
\end{proof}

\section{Proof of Theorem~\ref{thm:main}}\label{sec:main-result}

The proof follows the program developed in~\cite{Gross09} which itself adapted the strategy proposed in~\cite{CandesRecht08}.   The main idea is to approximate a dual feasible solution of~\eq{sdp} which certifies that $\mtx{M}$ is the unique minimum nuclear norm solution.   In~\cite{CandesRecht08} such a certificate was constructed via an infinite series using a construction developed in the compressed sensing literature~\cite{crt06,Fuchs04}.  The terms in this series were then analyzed individually using the decoupling inequalities of de la Pe\~{n}a and Montgomery-Smith~\cite{delaPena2}.  Truncating the infinite series after $4$ terms gave their result.  In~\cite{CandesTao09}, the authors bounded the contribution of $O(\log(n_2))$ terms in this series using intensive combinatorial analysis of each term.  The insight in~\cite{Gross09} was that, when sampling observations with replacement, a dual feasible solution could be closely approximated by a modified series where each term involved the product of independent random variables.   This change in the sampling model allows one to avoid decoupling inequalities and gives rise to the dramatic simplification here.

To proceed, recall again that by Proposition~\ref{prop:swr} it suffices to consider the scenario when the entries are sampled independently and uniformly with replacement.  I will first develop the main argument of the proof assuming many conditions hold with high probability. The proof is completed by subsequently bounding probability that all of these events hold.
Suppose that
\begin{equation}~\label{eq:big-set-isometry}
\frac{n_1n_2}{m} \left\|\PT \RO \PT - \frac{m}{n_1n_2}\PT\right\| \leq \frac{1}{2}\,, \quad\qquad \|\RO\| \leq \tfrac{8}{3} \beta^{1/2} \log(n_2)\,.
\end{equation}
Also suppose there exists a $\mtx{Y}$ in the range of $\RO$ such that
\begin{equation}\label{eq:quasi-multiplier}
	\|\PT(\mtx{Y})-\mtx{U}\mtx{V}^*\|_F \leq \sqrt{\frac{r}{2n_2}}\, , \quad\qquad \|\PTc(\mtx{Y})\| < \frac{1}{2}
\end{equation}

If~\eq{big-set-isometry} holds, then for any $\mtx{Z} \in \ker{\RO}$, $\PT(\mtx{Z})$ cannot be too large.  Indeed, we have
\[
	0 =\|\RO(\mtx{Z})\|_F \geq \|\RO\PT(\mtx{Z})\|_F - \|\RO\PTc(\mtx{Z})\|_F\,.
\]
Now observe that 
\[
 \|\RO\PT(\mtx{Z})\|_F^2 = \langle \mtx{Z}, \PT\RO^2\PT(\mtx{Z})\rangle \geq  \langle \mtx{Z}, \PT\RO\PT(\mtx{Z})\rangle \geq \frac{m}{2n_1n_2} \|\PT(\mtx{Z})\|_F^2
\]
and $\|\RO\PTc(\mtx{Z})\|_F\leq \tfrac{8}{3}\beta^{1/2}\log(n_2) \|\PTc(\mtx{Z})\|_F$.  Collecting these facts gives that for any $\mtx{Z} \in \ker{\RO}$,  
\[
\|\PTc(\mtx{Z})\|_F \geq  \sqrt{\frac{9m}{128 \beta n_1n_2\log^2(n_2)}} \|\PT(\mtx{Z})\|_F>\sqrt{\frac{2r}{n_2}} \|\PT(\mtx{Z})\|_F\,.
\]  
Now recall that $\|\mtx{A}\|_* = \sup_{\|B\|\leq 1} \langle \mtx{A},\mtx{B}\rangle$.  For $\mtx{Z}\in \ker \RO$, pick $\mtx{U}_\perp$ and $\mtx{V}_\perp$ such that $[\mtx{U},\mtx{U}_\perp]$ and $[\mtx{V},\mtx{V}_\perp]$ are unitary matrices and that $\langle \mtx{U}_\perp\mtx{V}_\perp^*, \PTc(\mtx{Z}) \rangle = \|\PTc(\mtx{Z})\|_*$.   Then it follows that 
\begin{align*}
	\|\mtx{M} + \mtx{Z}\|_* &\geq \langle \mtx{U}\mtx{V}^* + \mtx{U}_\perp \mtx{V}_\perp^*,
	\mtx{M} + \mtx{Z} \rangle\\
	&= \|\mtx{M}\|_* + \langle \mtx{U}\mtx{V}^* + \mtx{U}_\perp\mtx{V}_\perp^*, \mtx{Z} \rangle\\
	&=  \|\mtx{M}\|_* + \langle \mtx{U}\mtx{V}^* -\PT(\mtx{Y}), \PT( \mtx{Z} )\rangle 
		+\langle \mtx{U}_\perp \mtx{V}_\perp^*-\PTc(\mtx{Y}), \PTc(\mtx{Z}) \rangle\\
		& > \|\mtx{M}\|_* -  \sqrt{\frac{r}{2n_2}}\|\PT(\mtx{Z})\|_F +  \frac{1}{2}\|\PTc(\mtx{Z})\|_*
			\geq \|\mtx{M}\|_*\,.
\end{align*}
The first inequality holds from the variational characterization of the nuclear norm.  We also used the fact that $\langle \mtx{Y},\mtx{Z} \rangle =0$ for all $\mtx{Z}\in\ker{\RO}$.  Thus, if a $\mtx{Y}$ exists obeying~\eq{quasi-multiplier}, we have that for any $\mtx{X}$ obeying $\RO(\mtx{X}-\mtx{M})=\mtx{0}$, $\|\mtx{X}\|_* > \|\mtx{M}\|_*$.  That is, any if $\mtx{X}$ has $M_{ab}=X_{ab}$ for all $(a,b)\in\Omega$, $\mtx{X}$ has strictly larger nuclear norm than $\mtx{M}$, and hence $\mtx{M}$ is the unique minimizer of~\eq{sdp}.  The remainder of the proof shows that such a $\mtx{Y}$ exists with high probability.

To this end, partition $1,\ldots, m$ into $p$ partitions of size $q$. By assumption, we may choose
\[
q\geq\frac{128}{3} \max\{\mu_0,\mu_1^2\} r (n_1+n_2) \beta \log(n_1+n_2)\quad \mbox{and} \quad p\geq\frac{3}{4} \log(2n_2)\,.
\]
 Let $\Omega_j$ denote the set of indices corresponding to the $j$th partition.  Note that each of these partitions are independent of one another when the indices are sampled with replacement.  Assume that 
\begin{equation}~\label{eq:small-set-isometry}
\frac{n_1n_2}{q} \left\|\PT \ROs{k} \PT - \frac{q}{n_1n_2}\PT\right\| \leq \frac{1}{2}
\end{equation}
for all $k$. Define $\mtx{W}_0 = \mtx{U}\mtx{V}^*$ and set $\mtx{Y}_k = \frac{n_1n_2}{q}\sum_{j=1}^k \ROs{j}(\mtx{W}_{j-1})$, $\mtx{W}_k = \mtx{U}\mtx{V}^* - \PT(\mtx{Y}_k)$ for $k=1,\ldots, p$.  Then
\[
	\|\mtx{W}_k\|_F = \left\|\mtx{W}_{k-1} - \frac{n_1n_2}{q}\PT\ROs{k}(\mtx{W}_{k-1})\right\|_F 
	= \left\|(\PT-\frac{n_1n_2}{q}\PT\ROs{k}\PT)(\mtx{W}_{k-1})\right\|_F
	\leq \frac{1}{2} \|\mtx{W}_{k-1}\|_F\,,
\]
and it follows that $\|\mtx{W}_k\|_F \leq 2^{-k}\|\mtx{W}_0\|_F = 2^{-k}\sqrt{r}$. Since $p\geq\tfrac{3}{4} \log(2n_2) \geq \tfrac{1}{2}\log_2(2n_2) = \log_2 \sqrt{2n_2}$, then $\mtx{Y}=\mtx{Y}_p$ will satisfy the first inequality of~\eq{quasi-multiplier}.  Also suppose that
\begin{align}\label{eq:w-norm-shrink1}
\left\|\mtx{W}_{k-1} - \frac{n_1n_2}{q}\PT\ROs{k}(\mtx{W}_{k-1})\right\|_\infty 
	&\leq \frac{1}{2} \|\mtx{W}_{k-1}\|_\infty\\
	\left\|\left(\frac{n_1n_2}{q}\ROs{j}-\OpId\right)(\mtx{W}_{j-1})\right\| &\leq \sqrt{\frac{8 n_1 n_2^2 \beta  \log n_2}{3q}} \|\mtx{W}_{j-1}\|_\infty
	\label{eq:w-norm-shrink2}
\end{align}
for $k=1,\ldots,p$.

To see that $\|\PTc(\mtx{Y}_p)\|\leq \frac{1}{2}$ when~\eq{w-norm-shrink1} and~\eq{w-norm-shrink2} hold, observe
$\|\mtx{W}_k\|_\infty \leq 2^{-k} \|\mtx{U}\mtx{V}^*\|_\infty$, and it follows that
\begin{align*}
	\|\PTc \mtx{Y}_p\| &\leq \sum_{j=1}^p \|\tfrac{n_1n_2}{q}\PTc \ROs{j}\mtx{W}_{j-1}\|\\
	 &=  \sum_{j=1}^p \|\PTc( \tfrac{n_1n_2}{q}\ROs{j}\mtx{W}_{j-1}-\mtx{W}_{j-1})\|\\
	 &\leq  \sum_{j=1}^p \|(\tfrac{n_1n_2}{q} \ROs{j} -\OpId)(\mtx{W}_{j-1})\|\\
	 &\leq  \sum_{j=1}^p  \sqrt{\frac{8n_1n_2^2\, \beta\log n_2}{3q}}\| \mtx{W}_{j-1}\|_\infty\\
	  &=  2\sum_{j=1}^p 2^{-j} \sqrt{\frac{8n_1n_2^2 \, \beta \log n_2}{3q}}\|\mtx{U}\mtx{V}^*\|_\infty < 
	  \sqrt{\frac{32\mu_1^2 r n_2 \, \beta\log n_2}{3q}}<1/2
\end{align*}
since $q> \tfrac{128}{3} \mu_1^2 r n_2 \beta \log(n_2)$.  The first inequality follows from the triangle inequality.  The second line follows because $\mtx{W}_{j-1}\in T$ for all $j$.  The third line follows because, for any $\mtx{Z}$, 
\[
\|\PTc(\mtx{Z})\| = \|(\Id_{n_1} - \mtx{P}_{U})\mtx{Z}(\Id_{n_2} - \mtx{P}_{V})\| \leq \|\mtx{Z}\|\,.
\]
The fourth line applies~\eq{w-norm-shrink2}.  The  next line follows from~\eq{w-norm-shrink1}.  The final line follows from the assumption {\bf A1}.

All that remains is to bound the probability that all of the invoked events hold.  With $m$ satisfying the bound in the main theorem statement, the first inequality in \eq{big-set-isometry}~fails to hold with probability at most $2n_2^{2-2\beta}$ by Theorem~\ref{thm:near-isometry}, and the second inequality fails to hold with probability at most $n_2^{2-2\beta^{1/2}}$ by Proposition~\ref{prop:duplicate-count}.  For all $k$, \eq{small-set-isometry} fails to hold with probability at most $2n_2^{2-2\beta}$, \eq{w-norm-shrink1} fails to hold with probability at most $2n_2^{2-2\beta}$, and \eq{w-norm-shrink2} fails to hold with probability at most $(n_1+n_2)^{1-2\beta}$.  Summing these all together, all of the events hold with probability at least 
\[
	1-6\log(n_2) (n_1+n_2)^{2-2\beta}-n_2^{2-2\beta^{1/2}}
\] by the union bound.  This completes the proof.

\section{Discussion and Conclusions}\label{sec:conclusions}

The results proven here are nearly optimal, but small improvements can possibly be made.  The numerical constant $32$ in the statement of the theorem may be reducible by more clever bookkeeping, and it may be possible to derive a linear dependence on the logarithm of the matrix dimensions.  But further reduction is not possible because of the necessary conditions provided by Cand\`es and Tao.  One minor improvement that could be made would be to remove the assumption {\bf A1}.  For instance, while $\mu_1$ is known to be small in most of the models of low rank matrices that have been analyzed, no one has shown that an assumption of the form {\bf A1}  is necessary for completion.  Nonetheless, all prior results on matrix completion have imposed an assumption like {\bf A1}~\cite{CandesRecht08, CandesTao09,Keshavan09}, and it would be interesting to see if it can be removed as a requirement, or if it is somehow necessary. 

\if{0}
While it was known that the nuclear norm problem could be efficiently solved by semidefinite programming~\cite{FazelThesis, Fazel01, Vandenberghe96}, the results of~\cite{Recht07,CandesRecht08, CandesTao09, Keshavan09} have inspired the development of many special purpose algorithms to rapidly minimize the nuclear norm~\cite{Cai08,Ji09, Liu08, Ma08}.  
\fi

Surprisingly, the simplicity of the argument presented here mostly arises from the abandonment of Bernoulli sampling in favor of sampling with replacement.  It would be of interest to review results investigating noise robustness of matrix completion~\cite{CandesPlan09,Keshavan09b} or deconvolution of sparse and low rank matrices~\cite{Chandrasekaran09} to see if results can be improved by appealing to sampling with replacement.  Furthermore, since much of the work on rank minimization and matrix completion borrows tools from the compressed sensing community, it is of interest to revisit this related body of work and to see if proofs can be simplified or bounds can be improved there as well.  The noncommutative versions of Chernoff and Bernstein' s Inequalities may be useful throughout machine learning and statistical signal processing, and a fruitful line of inquiry would examine how to apply these tools from quantum information to the study of classical signals and systems.

\subsection*{Acknowledgments}
B.R. would like to thank Aram Harrow for introducing him to the operator Chernoff bound and many helpful clarifying conversations, Silvia Gandy for pointing out several typos in the original version of this manuscript, and Rob Nowak, Ali Rahimi, and Stephen Wright for many fruitful discussions about this paper.

\begin{small}
\bibliographystyle{abbrv}
\bibliography{/Users/brecht/Documents/LaTeX/brecht}
\end{small}

\appendix
\section{Operator Chernoff Bounds}

In this section, I present a proof of~\ref{thm:bernstein}, and also provide new proofs of some probability bounds from quantum information theory.  To review, a symmetric matrix $\mtx{A}$ is positive semidefinite if all of its eigenvalues are nonnegative. If $\mtx{A}$ and $\mtx{B}$ are positive semidefinite matrices, $\mtx{A}\preceq \mtx{B}$ means $\mtx{B}-\mtx{A}$ is positive semidefinite.  For square matrices $\mtx{A}$, the matrix exponential will be denoted $\exp(\mtx{A})$ and is given by the power series
\[
	\exp(\mtx{A}) = \sum_{k=0}^\infty \frac{\mtx{A}^k}{k!}
\]

The following theorem is a generalization of Markov's inequality originally proven in~\cite{AhlswedeWinter02}.  My proof closely follows the standard proof of the traditional Markov inequality, and does not rely on discrete summations.
\begin{theorem}[Operator Markov Inequality~\cite{AhlswedeWinter02}]\label{thm:markov}
	Let $\mtx{X}$ be a random positive semidefinite matrix and $\mtx{A}$ a fixed positive definite matrix.  Then
\[
		\P\left[ \mtx{X} \not\preceq \mtx{A}\right] \leq \trace(\E[\mtx{X}]\mtx{A}^{-1})
\]
\end{theorem}
\begin{proof}
Note that if $\mtx{X}\not\preceq\mtx{A}$, then $\mtx{A}^{-1/2}\mtx{X}\mtx{A}^{-1/2} \not\preceq \mtx{I}$, and hence $\|\mtx{A}^{-1/2}\mtx{X}\mtx{A}^{-1/2}\|>1$.  Let $I_{\mtx{X}\not\preceq\mtx{A}}$ denote the indicator of the event $\mtx{X}\not\preceq\mtx{A}$.  Then $I_{\mtx{X}\not\preceq\mtx{A}} \leq \trace(\mtx{A}^{-1/2}\mtx{X}\mtx{A}^{-1/2})$ as the right hand side is always nonnegative, and, if the left hand side equals $1$, the trace of the right hand side must exceed the norm of the right hand side which is greater than $1$. Thus we have
\[
\P[\mtx{X}\not\preceq\mtx{A}] = \E[I_{\mtx{X}\not\preceq\mtx{A}}] 
\leq\E[\trace(\mtx{A}^{-1/2}\mtx{X}\mtx{A}^{-1/2})]
= \trace(\E[\mtx{X}]\mtx{A}^{-1})\,.
\]
where the last equality follows from the linearity and cyclic properties of the trace.
\end{proof}

Next I will derive a noncommutative version of the Chernoff bound.  This was also proven in ~\cite{AhlswedeWinter02} for i.i.d. matrices. The version stated here is more general in that the random matrices need not be identically distributed, but the proof is essentially the same.  

\begin{theorem}[Noncommutative Chernoff Bound]\label{thm:chernoff}
Let $\mtx{X}_1,\ldots,\mtx{X}_n$ be independent symmetric random matrices in $\R^{d\times d}$.  Let $\mtx{A}$ be an arbitrary symmetric matrix.  Then for any invertible $d\times d$ matrix $\mtx{T}$
\[
	\P\left[ \sum_{k=1}^n \mtx{X}_k \not\preceq n\mtx{A} \right] \leq d \prod_{k=1}^n\left\| \E[\exp(
	\mtx{T}\mtx{X_k}\mtx{T}^* - \mtx{T}\mtx{A}\mtx{T}^*)] \right\|
\]
 \end{theorem}
 \begin{proof}
The proof relies on an estimate from statistical physics which is stated here without proof.
\begin{lemma}[Golden-Thompson inequality~\cite{Golden65,Thompson65}]
For any symmetric matrices $\mtx{A}$ and $\mtx{B}$, 
\[
\trace(\exp(\mtx{A}+\mtx{B})) \leq \trace((\exp \mtx{A}) (\exp \mtx{B}))\,.
\]
\end{lemma}
Much like the proof of the standard Chernoff bound, the theorem now follows from a long chain of inequalities.  
 \begin{align*}
 	\P\left[ \sum_{k=1}^n \mtx{X}_k \not\preceq n\mtx{A} \right]&=
 	\P\left[ \sum_{k=1}^n (\mtx{X}_k - \mtx{A})\not\preceq 0 \right]\\
	&= \P\left[ \sum_{k=1}^n\mtx{T} (\mtx{X}_k-\mtx{A}) \mtx{T}^*  \not\preceq 0 \right]\\
	&= \P\left[ \exp\left(\sum_{k=1}^n\mtx{T}( \mtx{X}_k-\mtx{A}) \mtx{T}^* \right)\not\preceq \mtx{I}_d \right]\\
	&\leq \trace\left( \E\left[ \exp\left(\sum_{k=1}^n\mtx{T}( \mtx{X}_k-\mtx{A}) \mtx{T}^*\right) \right] \right)\\
	&= \E\left[ \trace\left(  \exp\left(\sum_{k=1}^n\mtx{T}(\mtx{X}_k - \mtx{A}) \mtx{T}^*  \right) \right)  \right]\\
	&\leq \E\left[ \trace\left(  \exp\left(\sum_{k=1}^{n-1}\mtx{T}(\mtx{X}_k - \mtx{A}) \mtx{T}^*  \right) 
	\exp\left(\mtx{T}(\mtx{X}_n - \mtx{A}) \mtx{T}^*  \right) 
	\right)  \right]\\
	&\leq \E_{1,\ldots,{n-1}}\left[ \trace\left(  \exp\left(\sum_{k=1}^{n-1}\mtx{T}(\mtx{X}_k - \mtx{A}) \mtx{T}^*  \right) 
	\E[\exp\left(\mtx{T}(\mtx{X}_n - \mtx{A}) \mtx{T}^*  \right)]
	\right)  \right]\\
	&\leq \left\|\E[\exp\left(\mtx{T}(\mtx{X}_n - \mtx{A}) \mtx{T}^*  \right)]\right\|\E_{1,\ldots,{n-1}}\left[ \trace\left(  \exp\left(\sum_{k=1}^{n-1}\mtx{T}(\mtx{X}_k - \mtx{A}) \mtx{T}^*  \right) \right)  \right]\\
	&\leq \prod_{k=2}^n\left\|\E[\exp\left(\mtx{T}(\mtx{X}_k - \mtx{A}) \mtx{T}^*  \right)]\right\|\E\left[ \trace\left(  \exp\left(\mtx{T}(\mtx{X}_1 - \mtx{A}) \mtx{T}^*  \right) \right)  \right]\\
		&\leq d\prod_{k=1}^n\left\|\E[\exp\left(\mtx{T}(\mtx{X}_k - \mtx{A}) \mtx{T}^*  \right)]\right\|
 \end{align*}
Here, the first three lines follow from standard properties of the semidefinite ordering.  The fourth line invokes the Operator Markov Inequality.  The sixth line follows from the Golden-Thompson inequality. The seventh line follows from independence of the $\mtx{X}_k$.  The eighth line follows because for positive definite matrices $\trace(\mtx{A}\mtx{B}) \leq \trace(\mtx{A})\|\mtx{B}\|$.  This is just another statement of the duality between the nuclear and operator norms.  The ninth line iteratively repeats the previous two steps.  The final line follows because for a positive definite matrix $\mtx{A}$, $\trace(\mtx{A})$ is the sum of the eigenvalues of $\mtx{A}$, and all of the eigenvalues are at most $\|\mtx{A}\|$.
 \end{proof}

Let us now turn to proving the Noncommutative Bernstein Inequality presented in Section~\ref{sec:swr}. The authors in~\cite{Gross09} proposed a similar inequality for symmetric i.i.d. random matrices with a slightly worse constant.  The proof here is more general and follows the standard derivation of Bernstein's inequality.

\begin{proof}[of Theorem~\ref{thm:bernstein}]
Set 
\[
	\mtx{Y}_k = \left[\begin{array}{cc} \mtx{0} & \mtx{X}_k\\ \mtx{X}_k^* & \mtx{0}\end{array}\right]
\]
Then $\mtx{Y}_k$ are symmetric random variables, and for all $k$
\[
	\|\E[\mtx{Y}_k^2]\| = \left\| \E\left[ \left[ \begin{array}{cc} \mtx{X}_k\mtx{X}_k^* & \mtx{0} \\
	\mtx{0} & \mtx{X}_k^*\mtx{X}_k \end{array}\right]\right] \right\| = \max \{
	\| \E[  \mtx{X}_k\mtx{X}_k^*]\|, \| \E[\mtx{X}_k^*\mtx{X}_k]\|\} = \rho_k^2\,.
\]
Moreover, the maximum singular value of $\sum_{k=1}^L \mtx{X}_k$ is equal to the maximum eigenvalue of $\sum_{k=1}^L \mtx{Y}_k$. By Theorem~\ref{thm:chernoff}, we have for all $\lambda>0$
\[
	\P\left[\left\| \sum_{k=1}^L \mtx{X}_k \right\| > Lt \right] =
	\P\left[ \sum_{k=1}^L \mtx{Y}_k \not\preceq Lt\mtx{I} \right] \leq (d_1+d_2)\exp(-L\lambda t) \prod_{k=1}^L	\left\| \E[\exp(\lambda\mtx{Y_k})] \right\|\,.
\]

For each $k$, let $\mtx{Y}_k = \mtx{U}_k \mtx{\Lambda}_k \mtx{U}_k^*$ be an eigenvalue decomposition, where $\mtx{\Lambda}_k$ is the diagonal matrix of the eigenvalues of $\mtx{Y}_k$.  In turn, it follows that for $s>0$
\[
-M^s\mtx{Y}_k^2 \preceq -\mtx{U}_k M^s\mtx{\Lambda}_k^2 \mtx{U}_k^* \preceq
\mtx{U}_k \mtx{\Lambda}_k^{2+s} \mtx{U}_k^* = \mtx{Y}_k^{2+s} \preceq \mtx{U}_k M^s\mtx{\Lambda}_k^2 \mtx{U}_k^* \preceq M^s \mtx{Y}_k^2\,,
\]
which then implies
\begin{equation}\label{eq:pow-norm-bound}
	\|\E[\mtx{Y}_k^{s+2}]\|\leq M^s\|\E[\mtx{Y}_k^2]\|\,.
\end{equation}

For fixed $k$, we have
\begin{align*}
	\|\E[\exp(\lambda \mtx{Y}_k)]\| 
	&\leq \|\mtx{I}\| + \sum_{j=2}^\infty \frac{\lambda^j}{j!}\|\E[\mtx{Y}_k^j]\|\\
	&\leq  1 + \sum_{j=2}^\infty \frac{\lambda^j}{j!}\|\E[\mtx{Y}_k^2]\|M^{j-2}\\
	&=  1 + \frac{\rho_k^2}{M^2}\sum_{j=2}^\infty \frac{\lambda^j}{j!}M^{j}
	=  1 + \frac{\rho_k^2}{M^2}(\exp(\lambda M)-1 - \lambda M)\\
	&\leq  \exp\left(\frac{\rho_k^2}{M^2}(\exp(\lambda M)-1 - \lambda M)\right)\,.
\end{align*}

The first inequality follows from the triangle inequality and the fact that $\E[\mtx{Y}_k]=0$, the second inequality follows from (\ref{eq:pow-norm-bound}), and the final inequality follows from the fact that $1+x \leq \exp(x)$ for all $x$. Putting this together gives
\[
	\P\left[\left\| \sum_{k=1}^L \mtx{X}_k \right\| > Lt \right]  \leq (d_1+d_2) \exp
	\left(-\lambda Lt  + \frac{\sum_{k=1}^L\rho_k^2}{M^2}(\exp(\lambda M)-1 - \lambda M)\right)\,.
	\]
This final expression is now just a real number, and only has to be minimized as a function of $\lambda$.  The theorem now follows by algebraic manipulation: the right hand side is minimized by setting $\lambda = \frac{1}{M}\log(1+\frac{tLM}{\sum_{k=1}^L \rho_k^2})$, then basic approximations can be employed to complete the argument (see, for example~\cite{MIT18465}, lectures 4 and 5).
\end{proof}

\end{document}